\documentclass[aps,pra,twocolumn,superscriptaddress,nofootinbib,floatfix]{revtex4-2}

\usepackage[T1]{fontenc}
\usepackage{lmodern}
\usepackage{amsmath,amssymb,mathtools,bm}
\usepackage{amsthm}
\usepackage{booktabs,array}
\usepackage{xcolor}
\usepackage{tikz}
\usepackage{pgfplots}
\pgfplotsset{compat=1.18}
\usepackage{hyperref}
\hypersetup{
  colorlinks=true,
  citecolor=blue,
  linkcolor=blue,
  urlcolor=blue,
  pdftitle={Coupling Structure Determines Communication Regimes in Discrete Quantum Convolutional Channels},
  pdfauthor={Chunhe Xiong}
}

\newtheorem{theorem}{Theorem}
\newtheorem{proposition}[theorem]{Proposition}
\newtheorem{lemma}[theorem]{Lemma}

\newtheorem{remark}[theorem]{Remark}

\newcommand{\Fd}{\mathbb F_d}
\newcommand{\cH}{\mathcal H}
\newcommand{\States}{\mathsf D}
\newcommand{\id}{\operatorname{id}}
\newcommand{\Tr}{\operatorname{Tr}}
\newcommand{\Smin}{S_{\min}}
\newcommand{\Sreg}{\overline S_{\min}}
\newcommand{\CZ}{C_{\mathrm{rel}}^{Z}}
\newcommand{\CF}{C_{\mathrm{rel}}^{F}}
\newcommand{\ket}[1]{\lvert #1\rangle}
\newcommand{\bra}[1]{\langle #1\rvert}
\newcommand{\proj}[1]{\lvert #1\rangle\!\langle #1\rvert}

\begin{document}

\title{Coupling Structure Determines Communication Regimes in Discrete Quantum Convolutional Channels}

\author{Chunhe Xiong}
\email{xiongchunhe@csu.edu.cn}
\affiliation{School of Mathematics and Statistics, Central South University, Changsha 410083, China}


\begin{abstract}

We study the recently developed theory of quantum convolution
for discrete-variable quantum systems. We classify the discrete
quantum convolutional channels generated by an invertible
\(2\times2\) coupling matrix and a fixed environmental state
according to which entries of the coupling matrix vanish.
Under role-preserving basis relabelings, matrices with all four
entries nonzero form \(d-2\) canonical cross-ratio classes,
whereas matrices with exactly one vanishing entry split into
four inequivalent branches. We determine the optimized one-shot
Holevo information and the classical, quantum, and private
capacities for all four one-entry-vanishing branches. If a
diagonal entry vanishes, the channel is an entanglement-breaking
measure-and-prepare orbit channel: its classical capacity equals
a relative-entropy coherence of the environmental state, whereas
its quantum and private capacities vanish. If an off-diagonal
entry vanishes, the channel is a degradable generalized-dephasing
channel: one complete orthonormal basis is transmitted without
error, whereas the quantum and private capacities are determined
by the entropy deficit of the environmental state after
dephasing in the conjugate basis. Thus, the position of a single
vanishing entry selects two qualitatively different communication
regimes. When all four entries are nonzero, the channels are
irreducibly Weyl covariant, reducing the classical-capacity
problem to a regularized minimum-output-entropy problem. For the
unique fully nonzero coupling class in the single-qutrit case,
we further derive an exact Holevo formula for a depolarized
one-parameter family and identify its entanglement-breaking
threshold.

\end{abstract}

\maketitle

\section{Introduction}

The classical, quantum, and private capacities of a quantum channel quantify the reliable rates for transmitting classical information, quantum information, and secret classical information, respectively. The classical capacity is the regularization of the optimized one-shot Holevo information \cite{Holevo1998,SchumacherWestmoreland1997}; the quantum and private capacities are regularizations of coherent-information and private-information quantities \cite{Lloyd1997,Devetak2005,DevetakShor2005}. Their exact evaluation is difficult because additivity can fail even for channels with simple descriptions \cite{Hastings2009,Leditzky2018}. Structural classifications are therefore useful when they convert regularized optimizations into explicit formulas.

Discrete quantum convolution provides a finite-dimensional setting in which a system input and a fixed environmental state are mixed by an invertible linear transformation over a finite field \cite{BuGuJaffePNAS,BuGuJaffeDQG}. The resulting channels depend on two independent ingredients: the environmental state and the algebraic pattern of the coupling coefficients. Previous work concentrated mainly on couplings whose four entries are nonzero and related their communication properties to minimum output entropy, stabilizer structure, and environmental nonstabilizerness \cite{XiongMean2026,BuJaffe2025}. Here we show that removing a single coupling term does not merely produce a perturbative limit of that sector. Its position selects one of two qualitatively different channel structures: an entanglement-breaking orbit channel or a degradable generalized-dephasing channel.

We classify invertible couplings in the nontrivial sector consisting of fully nonzero matrices and matrices with exactly one vanishing entry. Following the convolution literature, we also call a fully nonzero coupling \emph{fully positive}; no order structure on the finite field is intended. The equivalence relation permits multiplicative relabelings of the system input and both outputs while leaving the environment input, and hence its state, fixed. Fully positive matrices reduce to $d-2$ canonical classes labeled by a cross-ratio. The one-entry-vanishing sector consists of four inequivalent branches, distinguished by the position of the zero.

For a vanishing diagonal entry, the channel measures the input in the computational or Fourier basis and prepares a complete shift or phase orbit of the environmental state. Such channels are entanglement breaking for every environment. Their optimized Holevo information and classical capacity equal the relative-entropy coherence of the environment in the conjugate basis, while their quantum and private capacities vanish. For a vanishing off-diagonal entry, one complete basis is transmitted without error. The complementary channel records only the corresponding classical basis label, placing the channel in the generalized-dephasing class and making it degradable \cite{HolevoComplementary,CubittRuskaiSmith}. Concavity of coherent information together with Weyl symmetry then identifies the maximally mixed state as optimal and gives exact single-letter formulas for both the quantum and private capacities. Thus the pattern of the missing coupling coefficient determines the induced channel class and all three capacities throughout the one-entry-vanishing sector.

The fully positive sector has a different structure. Weyl covariance reduces the optimized Holevo information to minimum output entropy and the regularized classical capacity to regularized minimum output entropy, but basis coherence alone does not determine either quantity. As an application, we study a depolarized single-qutrit slice. Its pure endpoint belongs to the balanced nonstabilizer family analyzed in Ref.~\cite{XiongMean2026}; the noisy interpolation considered here yields an exact all-parameter Holevo formula and an entanglement-breaking-to-NPT transition that are not part of that earlier analysis.

The remaining invertible matrices with two zero entries are diagonal or anti-diagonal. They induce, respectively, a unitary channel and a replacer channel, and are omitted as elementary endpoints. Section~\ref{sec:framework} introduces the model and proves the normal-form classification. Section~\ref{sec:onezero} determines all three capacities for the four one-entry-vanishing branches. Section~\ref{sec:fullpositive} treats the fully positive sector, and Sec.~\ref{sec:qutrit} analyzes the qutrit slice. Section~\ref{sec:discussion} concludes with the resulting capacity sectors and open additivity questions.

\section{Framework and coupling classification}
\label{sec:framework}

Throughout, $d$ is an odd prime, $V=\Fd^n$, and $D=|V|=d^n$. Let $\cH=(\mathbb C^d)^{\otimes n}$ with computational basis $\{\ket{x}:x\in V\}$. All logarithms are base two. For $x,p\in V$, define
\begin{equation*}
X^x\ket{y}=\ket{y+x},\qquad Z^p\ket{y}=\omega^{p\cdot y}\ket{y}.
\end{equation*}
Here $\omega=e^{2\pi i/d}$. Our Fourier convention is
\begin{equation}
\ket{\widetilde p}=\frac{1}{\sqrt D}\sum_{y\in V}\omega^{-p\cdot y}\ket{y},
\qquad X^x\ket{\widetilde p}=\omega^{p\cdot x}\ket{\widetilde p}.
\label{eq:fourier}
\end{equation}
The computational- and Fourier-basis dephasing maps are
\begin{align*}
\Delta_Z(\rho)&=\sum_{x\in V}\proj{x}\rho\proj{x},\\
\Delta_F(\rho)&=\sum_{p\in V}\proj{\widetilde p}\rho\proj{\widetilde p}.
\end{align*}
We use the relative-entropy coherences \cite{Baumgratz2014}
\begin{equation*}
\CZ(\rho)=S(\Delta_Z\rho)-S(\rho),
\CF(\rho)=S(\Delta_F\rho)-S(\rho).
\end{equation*}

Let
\begin{equation*}
M=\begin{pmatrix}a&b\\c&e\end{pmatrix}\in \mathrm{GL}(2,\Fd)
\end{equation*}
act componentwise on $V\oplus V$ through
\begin{equation}
U_M\ket{x,y}=\ket{ax+by,cx+ey}.
\label{eq:UM}
\end{equation}
For a fixed environmental state $\sigma\in\States(\cH)$, define
\begin{equation}
\Lambda_{M,\sigma}(\rho)=\Tr_2\!\left[U_M(\rho\otimes\sigma)U_M^\dagger\right].
\label{eq:channel}
\end{equation}
For $\alpha\in\Fd^\times$, let $P_\alpha\ket{x}=\ket{\alpha x}$.

For an ensemble $\mathcal E=\{p_j,\omega_j\}$, write
\begin{equation*}
\chi(\mathcal E):=S\!\left(\sum_jp_j\omega_j\right)-\sum_jp_jS(\omega_j).
\end{equation*}
For a channel $\Lambda$ and a complementary channel $\Lambda^c$, the optimized one-shot Holevo information and classical capacity are
\begin{align*}
\chi(\Lambda)&:=\max_{\{p_j,\rho_j\}}\chi\!\left(\{p_j,\Lambda(\rho_j)\}\right),\\
C(\Lambda)&:=\lim_{m\to\infty}\frac1m\chi(\Lambda^{\otimes m}).
\end{align*}
The coherent information and quantum capacity are
\begin{align*}
I_c(\rho,\Lambda)&:=S(\Lambda(\rho))-S(\Lambda^c(\rho)),\\
Q(\Lambda)&:=\lim_{m\to\infty}\frac1m
\max_{\rho_m}I_c(\rho_m,\Lambda^{\otimes m}).
\end{align*}
Finally, define the one-shot private information and private capacity by
\begin{align*}
P^{(1)}(\Lambda)&:=\max_{\{p_j,\rho_j\}}
\Bigl[\chi\!\left(\{p_j,\Lambda(\rho_j)\}\right)
-\chi\!\left(\{p_j,\Lambda^c(\rho_j)\}\right)\Bigr],\\
P(\Lambda)&:=\lim_{m\to\infty}\frac1mP^{(1)}(\Lambda^{\otimes m}).
\end{align*}
These are the standard regularized capacity formulas \cite{Holevo1998,SchumacherWestmoreland1997,Lloyd1997,Devetak2005,DevetakShor2005}.

\begin{theorem}[Normal-form classification]
\label{thm:classification}
Let $M\in\mathrm{GL}(2,\Fd)$ have at most one vanishing entry. Two coupling matrices are \emph{role-preserving multiplicatively equivalent} if
\begin{equation}
M'=\begin{pmatrix}\mu&0\\0&\nu\end{pmatrix}
M\begin{pmatrix}\alpha&0\\0&1\end{pmatrix},
\qquad \mu,\nu,\alpha\in\Fd^\times.
\label{eq:equiv}
\end{equation}
For fixed $\sigma$, equivalent matrices induce channels related by input and output permutation unitaries,
\begin{equation}
\Lambda_{M',\sigma}(\rho)
=P_\mu\Lambda_{M,\sigma}(P_\alpha\rho P_\alpha^\dagger)P_\mu^\dagger,
\label{eq:channel-equiv}
\end{equation}
and therefore have identical optimized Holevo information and identical classical, quantum, and private capacities.

If all four entries are nonzero, $M$ is equivalent to the unique normal form
\begin{equation}
M_r=\begin{pmatrix}r&1\\1&1\end{pmatrix},
\qquad r=\frac{ae}{bc}\in\Fd^\times\setminus\{1\}.
\label{eq:Mr}
\end{equation}
Hence the fully positive sector consists of $d-2$ equivalence classes.

If exactly one entry vanishes, $M$ is equivalent to exactly one of
\begin{equation}
\begin{pmatrix}0&1\\1&1\end{pmatrix},\quad
\begin{pmatrix}1&1\\1&0\end{pmatrix},\quad
\begin{pmatrix}1&0\\1&1\end{pmatrix},\quad
\begin{pmatrix}1&1\\0&1\end{pmatrix}.
\label{eq:fourforms}
\end{equation}
These four normal forms are inequivalent under Eq.~\eqref{eq:equiv}.
\end{theorem}

\begin{proof}
Equation~\eqref{eq:equiv} implies
\begin{equation*}
U_{M'}=(P_\mu\otimes P_\nu)U_M(P_\alpha\otimes I).
\end{equation*}
which gives Eq.~\eqref{eq:channel-equiv} after tracing out the second output. Unitary pre- and post-processing preserve the four quantities $\chi$, $C$, $Q$, and $P$.

For a fully positive matrix, choose $\mu=b^{-1}$, $\nu=e^{-1}$, and $\alpha=ec^{-1}$. Then
\begin{equation*}
\begin{pmatrix}b^{-1}&0\\0&e^{-1}\end{pmatrix}
M\begin{pmatrix}ec^{-1}&0\\0&1\end{pmatrix}
=\begin{pmatrix}ae/(bc)&1\\1&1\end{pmatrix}.
\end{equation*}
The ratio $ae/(bc)$ is invariant under Eq.~\eqref{eq:equiv}, proving uniqueness. Invertibility excludes $r=1$, so the $d-1$ nonzero field elements leave $d-2$ admissible classes.

When exactly one entry vanishes, the three nonzero entries can be normalized to one. Explicitly,
\begin{align*}
\begin{pmatrix}b^{-1}&0\\0&e^{-1}\end{pmatrix}
\begin{pmatrix}0&b\\c&e\end{pmatrix}
\begin{pmatrix}ec^{-1}&0\\0&1\end{pmatrix}
&=\begin{pmatrix}0&1\\1&1\end{pmatrix},\\
\begin{pmatrix}b^{-1}&0\\0&a(bc)^{-1}\end{pmatrix}
\begin{pmatrix}a&b\\c&0\end{pmatrix}
\begin{pmatrix}a^{-1}b&0\\0&1\end{pmatrix}
&=\begin{pmatrix}1&1\\1&0\end{pmatrix},\\
\begin{pmatrix}c(ae)^{-1}&0\\0&e^{-1}\end{pmatrix}
\begin{pmatrix}a&0\\c&e\end{pmatrix}
\begin{pmatrix}ec^{-1}&0\\0&1\end{pmatrix}
&=\begin{pmatrix}1&0\\1&1\end{pmatrix},\\
\begin{pmatrix}b^{-1}&0\\0&e^{-1}\end{pmatrix}
\begin{pmatrix}a&b\\0&e\end{pmatrix}
\begin{pmatrix}ba^{-1}&0\\0&1\end{pmatrix}
&=\begin{pmatrix}1&1\\0&1\end{pmatrix}.
\end{align*}
The transformations in Eq.~\eqref{eq:equiv} preserve the position of a zero, so the four normal forms are inequivalent.
\end{proof}

\begin{remark}[Cross-ratio branches and elementary endpoints]
\label{rem:branches}
Within the nontrivial sector of Theorem~\ref{thm:classification}, there are $d+2$ inequivalent classes: $d-2$ fully positive classes and four one-zero classes. The cross-ratio labels the fully positive sector, but the formal values $r=0$ and $r=\infty$ each split into two inequivalent zero-position branches. The omitted invertible two-zero matrices are diagonal or anti-diagonal; they give a unitary channel or a replacer channel, respectively. Thus the pattern of vanishing coupling coefficients, rather than the cross-ratio alone, determines whether the induced channel lies in an entanglement-breaking orbit sector, a degradable generalized-dephasing sector, or the fully positive Weyl-covariant sector.
\end{remark}

\section{Complete capacities with one vanishing entry}
\label{sec:onezero}

\subsection{One vanishing diagonal entry: orbit channels}

The following elementary orbit-channel lemma will be proved in Appendix~\ref{app:orbit}. If
\begin{equation}
\Phi_\tau(\rho)=\sum_{g\in\Gamma}\bra{g}\rho\ket{g}\,U_g\tau U_g^\dagger,
\label{eq:orbit-channel}
\end{equation}
where $g\mapsto U_g$ is a unitary representation of the finite group $\Gamma$, then \cite{Vaccaro2008pra,Gour2009pra,Marvian2016pra}
\begin{equation}
\chi(\Phi_\tau)=C(\Phi_\tau)
=S\!\left(\frac1{|\Gamma|}\sum_gU_g\tau U_g^\dagger\right)-S(\tau),
\label{eq:orbit-formula}
\end{equation}
and $Q(\Phi_\tau)=P(\Phi_\tau)=0$ \cite{ShorEBC,HorodeckiShorRuskai}.

\begin{theorem}[Diagonal-zero capacity formulas]
\label{thm:diagonalzero}
Let $M\in\mathrm{GL}(2,\Fd)$ have exactly one vanishing diagonal entry. By Theorem~\ref{thm:classification}, $M$ is equivalent to exactly one of
\begin{equation*}
M_{e=0}=\begin{pmatrix}1&1\\1&0\end{pmatrix},\qquad
M_{a=0}=\begin{pmatrix}0&1\\1&1\end{pmatrix}.
\end{equation*}
For the first class,
\begin{equation}
\chi=C=\CF(\sigma),\qquad Q=P=0,
\label{eq:e0-cap}
\end{equation}
whereas for the second class,
\begin{equation}
\chi=C=\CZ(\sigma),\qquad Q=P=0.
\label{eq:a0-cap}
\end{equation}
\end{theorem}

\begin{proof}
Write $\rho=\sum_{x,y}\rho_{xy}\ket{x}\bra{y}$ and $\sigma=\sum_{z,w}\sigma_{zw}\ket{z}\bra{w}$.

For $M_{e=0}$, one has $U_M\ket{x,z}=\ket{x+z,x}$. Tracing the second output forces the two input labels to agree, and therefore
\begin{equation*}
\Lambda_{M,\sigma}(\rho)
=\sum_{x\in V}\rho_{xx}X^x\sigma X^{-x}.
\end{equation*}
This is a computational-basis measurement followed by the full shift orbit of $\sigma$. Since the shift twirl equals $\Delta_F$, Eq.~\eqref{eq:orbit-formula} gives Eq.~\eqref{eq:e0-cap}.

For $M_{a=0}$, $U_M\ket{x,z}=\ket{z,x+z}$. Using
\begin{equation*}
\delta_{v,0}=\frac1D\sum_{m\in V}\omega^{m\cdot v}.
\end{equation*}
we obtain
\begin{align*}
\Lambda_{M,\sigma}(\rho)
&=\sum_{x,y,z,w}\rho_{xy}\sigma_{zw}
\delta_{x-y+z-w,0}\ket{z}\bra{w}\\
&=\frac1D\sum_{m\in V} A_m B_m,
\end{align*}
where
\begin{align*}
A_m&:=\sum_{x,y}\rho_{xy}\omega^{m\cdot(x-y)}
=D\bra{\widetilde m}\rho\ket{\widetilde m},\\
B_m&:=\sum_{z,w}\sigma_{zw}\omega^{m\cdot(z-w)}\ket{z}\bra{w}
=Z^m\sigma Z^{-m}.
\end{align*}
Consequently,
\begin{equation*}
\Lambda_{M,\sigma}(\rho)
=\sum_{m\in V}\bra{\widetilde m}\rho\ket{\widetilde m}\,
Z^m\sigma Z^{-m}.
\end{equation*}
Thus the channel measures the Fourier basis and prepares the full phase orbit. The phase twirl equals $\Delta_Z$, so Eq.~\eqref{eq:orbit-formula} gives Eq.~\eqref{eq:a0-cap}.
\end{proof}

\subsection{One vanishing off-diagonal entry: generalized dephasing}

A vanishing off-diagonal entry preserves one complete input basis. The complementary channel records only the corresponding classical basis probabilities, which makes the original channel degradable and its quantum and private capacities single-letter.

\begin{theorem}[Degradable channels with one vanishing off-diagonal coupling]
\label{thm:offdiag}
Let $M\in\mathrm{GL}(2,\Fd)$ have exactly one vanishing off-diagonal entry. For the representative
\begin{equation*}
M_{b=0}=\begin{pmatrix}1&0\\1&1\end{pmatrix},
\end{equation*}
the channel is degradable and
\begin{equation}
\chi=C=\log D,\qquad
Q=P=\log D-S(\Delta_F\sigma).
\label{eq:b0-cap}
\end{equation}
For the Fourier-dual representative
\begin{equation*}
M_{c=0}=\begin{pmatrix}1&1\\0&1\end{pmatrix},
\end{equation*}
the channel is also degradable and
\begin{equation}
\chi=C=\log D,\qquad
Q=P=\log D-S(\Delta_Z\sigma).
\label{eq:c0-cap}
\end{equation}
\end{theorem}

\begin{proof}
By Theorem~\ref{thm:classification} and unitary invariance of the four capacities, it suffices to analyze the two displayed representatives. We first consider $M_{b=0}$, for which
\begin{equation*}
U_M\ket{x,y}=\ket{x,x+y}.
\end{equation*}
Choose a purification $\ket{\Psi_\sigma}_{ER}$ of $\sigma_E$ and define
\begin{equation*}
\ket{\xi_x}_{ER}:=(X^x\otimes I_R)\ket{\Psi_\sigma}_{ER}.
\end{equation*}
The Stinespring isometry is
\begin{equation*}
V\ket{x}_A=\ket{x}_B\ket{\xi_x}_{ER}.
\end{equation*}
Consequently, for $\rho=\sum_{x,x'}\rho_{xx'}\ket{x}\bra{x'}$,
\begin{align}
\Lambda(\rho)
&=\sum_{x,x'}\rho_{xx'}\langle\xi_{x'}|\xi_x\rangle\ket{x}\bra{x'}\nonumber\\
&=\sum_{x,x'}\rho_{xx'}\Tr(\sigma X^{x-x'})\ket{x}\bra{x'},
\label{eq:b0-output}\\
\Lambda^c(\rho)
&=\sum_x\rho_{xx}\proj{\xi_x}.
\label{eq:b0-complement}
\end{align}
In particular, $\Lambda(\proj{x})=\proj{x}$ and $\langle x|\Lambda(\rho)|x\rangle=\rho_{xx}$. Define
\begin{equation*}
\mathcal G(\tau)=\sum_x\langle x|\tau|x\rangle\proj{\xi_x}.
\end{equation*}
Then $\mathcal G\circ\Lambda=\Lambda^c$, proving degradability. Moreover, the complementary channel is classical--quantum and hence entanglement breaking; this is the standard generalized-dephasing structure \cite{HolevoComplementary,CubittRuskaiSmith}.

We next identify the maximizing input for coherent information. For any ensemble $\{p_j,\rho_j\}$ with average $\bar\rho$, set $B_j=\Lambda(\rho_j)$ and $E_j=\Lambda^c(\rho_j)=\mathcal G(B_j)$. Then
\begin{align*}
&I_c(\bar\rho,\Lambda)-\sum_jp_jI_c(\rho_j,\Lambda)\\
&\qquad=\chi(\{p_j,B_j\})-\chi(\{p_j,E_j\})\ge0,
\end{align*}
where the inequality is the data-processing inequality for Holevo information. Thus $I_c(\rho,\Lambda)$ is concave \cite{YardHaydenDevetak,CubittRuskaiSmith}.

The symmetry needed for twirling follows directly from Eqs.~\eqref{eq:b0-output} and \eqref{eq:b0-complement}. For $t,s\in V$,
\begin{align*}
\Lambda(X^t\rho X^{-t})&=X^t\Lambda(\rho)X^{-t},\\
\Lambda^c(X^t\rho X^{-t})
&=(X^t\otimes I_R)\Lambda^c(\rho)(X^{-t}\otimes I_R),\\
\Lambda(Z^s\rho Z^{-s})&=Z^s\Lambda(\rho)Z^{-s},\\
\Lambda^c(Z^s\rho Z^{-s})&=\Lambda^c(\rho).
\end{align*}
Hence, for $W_{t,s}:=X^tZ^s$,
\begin{equation*}
I_c(W_{t,s}\rho W_{t,s}^\dagger,\Lambda)=I_c(\rho,\Lambda).
\end{equation*}
The complete Weyl twirl satisfies
\begin{equation*}
\frac1{D^2}\sum_{t,s\in V}W_{t,s}\rho W_{t,s}^\dagger=\frac ID.
\end{equation*}
Concavity and symmetry therefore imply
\begin{equation*}
I_c(I/D,\Lambda)\ge I_c(\rho,\Lambda)
\end{equation*}
for every $\rho$.

At the maximally mixed input, Eq.~\eqref{eq:b0-output} gives $\Lambda(I/D)=I/D$, so the receiver entropy is $\log D$. To compute the complementary entropy, expand
\begin{equation*}
\ket{\Psi_\sigma}_{ER}=\sum_{p\in V}\ket{\widetilde p}_E\ket{r_p}_R,
\end{equation*}
where the vectors $\ket{r_p}$ need not be normalized. Using Eq.~\eqref{eq:fourier} and character orthogonality,
\begin{align*}
\Lambda^c(I/D)
&=\frac1D\sum_x(X^x\otimes I_R)\proj{\Psi_\sigma}(X^{-x}\otimes I_R)\\
&=\sum_{p,q}\left(\frac1D\sum_x\omega^{(p-q)\cdot x}\right)
\ket{\widetilde p}\bra{\widetilde q}\otimes\ket{r_p}\bra{r_q}\\
&=\sum_p\proj{\widetilde p}\otimes\ket{r_p}\bra{r_p}.
\end{align*}
This state is block diagonal, and each nonzero block has the single eigenvalue
\begin{align*}
\langle\widetilde p|\sigma|\widetilde p\rangle=\langle\widetilde p|(\Tr_R\ket{\Psi_{\sigma}}\bra{\Psi_{\sigma}})|\widetilde p\rangle=\langle r_p|r_p\rangle.
\end{align*}
Its nonzero spectrum is therefore the spectrum of $\Delta_F\sigma$, and
\begin{equation*}
I_c(I/D,\Lambda)=\log D-S(\Delta_F\sigma).
\end{equation*}
For degradable channels, the optimized coherent information is additive, so the quantum capacity single-letters, and the private capacity equals the quantum capacity \cite{DevetakShor2005,CubittRuskaiSmith,Smith2008}. Hence the second formula in Eq.~\eqref{eq:b0-cap} follows. The uniform computational-basis ensemble is transmitted perfectly, so $\chi\ge\log D$. The output dimension gives the matching upper bound, proving $\chi=C=\log D$.

For $M_{c=0}$, $U_M\ket{x,y}=\ket{x+y,y}$. For every computational-basis environment vector $\ket y$,
\begin{align*}
U_M(\ket{\widetilde p}\otimes\ket y)
&=\frac1{\sqrt D}\sum_x\omega^{-p\cdot x}\ket{x+y,y}\\
&=\ket{\widetilde p}\otimes Z^p\ket y.
\end{align*}
By linearity, the corresponding Stinespring isometry satisfies
\begin{equation*}
V\ket{\widetilde p}_A
=\ket{\widetilde p}_B\,(Z^p\otimes I_R)\ket{\Psi_\sigma}_{ER}.
\end{equation*}
Thus the Fourier basis is transmitted perfectly and the complementary channel depends only on the Fourier-basis probabilities. Repeating the preceding argument with $X$ and $Z$, and with the computational and Fourier bases, interchanged gives Eq.~\eqref{eq:c0-cap}.
\end{proof}

Table~\ref{tab:capacities} summarizes the four one-entry-vanishing branches.

\begin{remark}[Communication mechanism]
The zero position determines which information reaches the receiver. A diagonal zero transfers an environmental orbit to the output and produces a measure-and-prepare channel. An off-diagonal zero leaves one basis label noiseless, while only the conjugate coherence leaks to the environment; this is the origin of degradability and of the entropy-deficit formulas in Theorem~\ref{thm:offdiag}.
\end{remark}

\begin{table*}[t]
\caption{Capacity formulas for all invertible couplings with exactly one vanishing entry. The distinguished basis is measured in the diagonal-zero sector and transmitted in the off-diagonal-zero sector.}
\label{tab:capacities}
\centering
\small
\renewcommand{\arraystretch}{1.22}
\begin{tabular}{@{}ccccc@{}}
\toprule
Zero & Basis action & Channel structure & $\chi=C$ & $Q=P$\\
\midrule
$e=0$ & computational basis measured & shift-orbit measure-and-prepare & $\CF(\sigma)$ & $0$\\
$a=0$ & Fourier basis measured & phase-orbit measure-and-prepare & $\CZ(\sigma)$ & $0$\\
$b=0$ & computational basis transmitted & generalized dephasing & $\log D$ & $\log D-S(\Delta_F\sigma)$\\
$c=0$ & Fourier basis transmitted & generalized dephasing & $\log D$ & $\log D-S(\Delta_Z\sigma)$\\
\bottomrule
\end{tabular}
\end{table*}

\section{Fully positive classes and minimum output entropy}
\label{sec:fullpositive}

By Theorem~\ref{thm:classification}, every fully positive coupling is equivalent to one of the $d-2$ matrices $M_r$ in Eq.~\eqref{eq:Mr}. These channels are Weyl covariant, but they are not generally orbit channels or generalized-dephasing channels.

\begin{proposition}[Weyl-covariant reduction]
\label{prop:weyl}
For a fully positive coupling $M$,
\begin{align}
\chi(\Lambda_{M,\sigma})&=\log D-\Smin(\Lambda_{M,\sigma}),\label{eq:chi-minout}\\
C(\Lambda_{M,\sigma})&=\log D-\Sreg(\Lambda_{M,\sigma}),
\label{eq:C-minout}
\end{align}
where
\begin{align*}
\Smin(\Phi)&=\min_\rho S(\Phi(\rho)),\\
\Sreg(\Phi)&=\lim_{m\to\infty}\frac1m\Smin(\Phi^{\otimes m}).
\end{align*}
\end{proposition}

\begin{proof}
Let $\delta=ae-bc$. The intertwining identities
\begin{align*}
U_M(X^t\otimes I)&=(X^{at}\otimes X^{ct})U_M,\\
U_M(Z^s\otimes I)&=(Z^{e\delta^{-1}s}\otimes Z^{-b\delta^{-1}s})U_M.
\end{align*}
show that conjugating the input by $X^t$ or $Z^s$ conjugates the receiver output by $X^{at}$ or $Z^{e\delta^{-1}s}$, respectively. Since $a,e\ne0$, these transformations generate the full output Weyl group. Let $\rho_*$ minimize the output entropy. The uniform Weyl orbit of $\rho_*$ has average input $I/D$, covariance makes all orbit outputs isospectral, and the average output is $I/D$. Its Holevo information is therefore $\log D-\Smin(\Lambda_{M,\sigma})$, which also saturates the general upper bound. Applying the same argument to tensor powers and regularizing gives Eq.~\eqref{eq:C-minout} \cite{AmosovHolevoWerner}.
\end{proof}

\begin{proposition}[Basis coherence is not a complete invariant]
\label{prop:notcoherence}
For every fully positive coupling $M$:
\begin{enumerate}
\item $\sigma=I/D$ gives the completely depolarizing channel and $C(\Lambda_{M,I/D})=0$;
\item $\CZ(\proj{0})=0$, but $C(\Lambda_{M,\proj{0}})=\log D$;
\item $\CF(\proj{\widetilde0})=0$, but $C(\Lambda_{M,\proj{\widetilde0}})=\log D$.
\end{enumerate}
Consequently, neither computational-basis nor Fourier-basis coherence alone determines the classical capacity in the fully positive sector.
\end{proposition}

\begin{proof}
For $\sigma=I/D$, expand $\rho=\sum_{x,x'}\rho_{xx'}\ket{x}\bra{x'}$. Tracing the second output gives
\begin{align*}
\Lambda_{M,I/D}(\rho)
&=\frac1D\sum_{x,x',y}\rho_{xx'}
\delta_{c(x-x'),0}\ket{ax+by}\bra{ax'+by}\\
&=\frac1D\sum_{x,y}\rho_{xx}\proj{ax+by}=\frac ID,
\end{align*}
where $b,c\ne0$ were used in the second equality. For $\sigma=\proj0$, Eq.~\eqref{eq:UM} maps $\ket{x,0}$ to $\ket{ax,cx}$, so the computational basis is transmitted perfectly. In the Fourier basis,
\begin{equation*}
U_M(\ket{\widetilde p}\otimes\ket{\widetilde q})
=\ket{\widetilde{\delta^{-1}(ep-cq)}}
\otimes\ket{\widetilde{\delta^{-1}(-bp+aq)}},
\end{equation*}
where $\delta=ae-bc$. Hence, for $\sigma=\proj{\widetilde0}$, the receiver obtains the distinct labels $\widetilde{e\delta^{-1}p}$, and a complete Fourier basis is transmitted perfectly. The latter two ensembles attain the output-dimension upper bound $\log D$.
\end{proof}

Thus the fully positive sector is the only nontrivial sector in the classification for which the general classical capacity remains a regularized minimum-output-entropy problem.

\section{An exact Holevo formula for a fully positive qutrit slice}
\label{sec:qutrit}

Set $d=3$ and $n=1$. The fully positive sector has a single class because the cross-ratio can only be $r=2$. We choose
\begin{equation*}
M_2=\begin{pmatrix}2&1\\1&1\end{pmatrix},\qquad
\ket{m}=\frac{\ket0+\ket1}{\sqrt2},
\end{equation*}
We then consider
\begin{equation}
\sigma_t=(1-t)\frac I3+t\proj m,\qquad 0\le t\le1.
\label{eq:sigmat}
\end{equation}
The pure endpoint $t=1$ is coupling-equivalent to the balanced qutrit point in Ref.~\cite{XiongMean2026}. The depolarized interpolation in Eq.~\eqref{eq:sigmat}, its all-parameter Holevo formula, and the threshold below are specific to the present analysis.

\begin{theorem}[Exact qutrit Holevo information and entanglement-breaking threshold]
\label{thm:qutrit}
For $\Lambda_{M_2,\sigma_t}$,
\begin{equation}
\chi(\Lambda_{M_2,\sigma_t})
=\log_2 3-H\!\left(\frac{2+t}{6},\frac{2+t}{6},\frac{1-t}{3}\right).
\label{eq:qutrit-chi}
\end{equation}
For $0\le t\le2/5$, the channel is entanglement breaking and
\begin{equation*}
C(\Lambda_{M_2,\sigma_t})=\chi(\Lambda_{M_2,\sigma_t}).
\end{equation*}
For $t>2/5$, its normalized Choi state is NPT, so the channel is not entanglement breaking.
\end{theorem}

\begin{proof}
For a pure input $\ket\psi$, let $\tau_\psi=\Lambda_{M_2,\proj m}(\proj\psi)$. Appendix~\ref{app:qutrit-entropy} proves that every $\tau_\psi$ has purity $1/2$ and that its entropy is minimized by the spectrum $(1/2,1/2,0)$, attained by computational-basis inputs. Since $\Lambda_{M_2,I/3}$ is completely depolarizing, linearity in the environment gives
\begin{equation*}
\Lambda_{M_2,\sigma_t}(\proj\psi)=\frac{1-t}{3}I+t\tau_\psi.
\end{equation*}
The minimum output spectrum is therefore
\begin{equation*}
\left(\frac{2+t}{6},\frac{2+t}{6},\frac{1-t}{3}\right),
\end{equation*}
and Proposition~\ref{prop:weyl} gives Eq.~\eqref{eq:qutrit-chi}.

At $t=2/5$,
\begin{align*}
\sigma_{2/5}&=\frac15\left(\proj0+\proj1+\sum_{k=0}^2\proj{\varphi_k}\right),\\
\ket{\varphi_k}&=\frac{\ket0+\ket1+\omega^k\ket2}{\sqrt3}.
\end{align*}
which is a convex mixture of qutrit stabilizer states. For $0\le t\le2/5$,
\begin{equation*}
\sigma_t=\left(1-\frac{5t}{2}\right)\frac I3+\frac{5t}{2}\sigma_{2/5}.
\end{equation*}
so $\sigma_t$ is also a stabilizer mixture. The stabilizer-environment theorem of Ref.~\cite{XiongPrivate2026} implies that a fully positive convolutional channel with such an environment is entanglement breaking. The convex decomposition above therefore covers the full interval $0\le t\le2/5$.

For the converse, Appendix~\ref{app:qutrit-choi} gives an explicit block diagonalization of the partially transposed normalized Choi state. Its eigenvalues are
\begin{equation*}
\frac{t+2}{18}\ (6\text{-fold}),\qquad
\frac{1+2t}{9},\qquad
\frac{2-5t}{18}\ (2\text{-fold}).
\end{equation*}
The last pair is negative exactly for $t>2/5$, proving the stated threshold.
\end{proof}

Figure~\ref{fig:qutrit} compares the exact one-shot Holevo information with the computational- and Fourier-basis coherences of the environmental state. The shaded region marks the entanglement-breaking interval \(0\le t\le 2/5\).

\begin{figure}[t]
\centering
\begin{tikzpicture}
\begin{axis}[
width=\columnwidth,height=5.5cm,
xlabel={$t$},ylabel={bits},
xmin=0,xmax=1,ymin=0,ymax=1.35,
legend style={font=\scriptsize,at={(0.02,0.98)},anchor=north west,draw=none,fill=white},
tick label style={font=\scriptsize},label style={font=\small},axis lines=left]
\addplot[draw=none,fill=gray!20] coordinates {(0,0) (0.4,0) (0.4,1.35) (0,1.35)} \closedcycle;
\addplot[black,thick] coordinates {(0.00,0.000000) (0.05,0.000909) (0.10,0.003672) (0.15,0.008343) (0.20,0.014988) (0.25,0.023684) (0.30,0.034519) (0.35,0.047595) (0.40,0.063034) (0.45,0.080981) (0.50,0.101607) (0.55,0.125122) (0.60,0.151786) (0.65,0.181926) (0.70,0.215967) (0.75,0.254479) (0.80,0.298270) (0.85,0.348566) (0.90,0.407454) (0.95,0.479338) (1.00,0.584963)};
\addlegendentry{$\chi$}
\addplot[black,dashed] coordinates {(0.00,0.000000) (0.05,0.002641) (0.10,0.010340) (0.15,0.022815) (0.20,0.039849) (0.25,0.061278) (0.30,0.086985) (0.35,0.116896) (0.40,0.150978) (0.45,0.189237) (0.50,0.231727) (0.55,0.278549) (0.60,0.329869) (0.65,0.385926) (0.70,0.447067) (0.75,0.513794) (0.80,0.586850) (0.85,0.667401) (0.90,0.757483) (0.95,0.861449) (1.00,1.000000)};
\addlegendentry{$\CZ(\sigma_t)$}
\addplot[black,dashdotted] coordinates {(0.00,0.000000) (0.05,0.002656) (0.10,0.010461) (0.15,0.023224) (0.20,0.040826) (0.25,0.063202) (0.30,0.090346) (0.35,0.122303) (0.40,0.159174) (0.45,0.201120) (0.50,0.248371) (0.55,0.301241) (0.60,0.360151) (0.65,0.425664) (0.70,0.498543) (0.75,0.579846) (0.80,0.671108) (0.85,0.774700) (0.90,0.894719) (0.95,1.039892) (1.00,1.251629)};
\addlegendentry{$\CF(\sigma_t)$}
\draw[densely dashed] (axis cs:0.4,0)--(axis cs:0.4,1.35);
\node[font=\scriptsize,anchor=south west] at (axis cs:0.41,1.17) {$t=2/5$};
\node[font=\scriptsize] at (axis cs:0.18,1.26) {EBC};
\end{axis}
\end{tikzpicture}
\caption{Exact one-shot Holevo information for the depolarized qutrit slice, compared with the two basis coherences. The shaded interval is entanglement breaking, where $C=\chi$. Beyond $t=2/5$, the solid curve remains the exact one-shot quantity; additivity is not asserted.}
\label{fig:qutrit}
\end{figure}
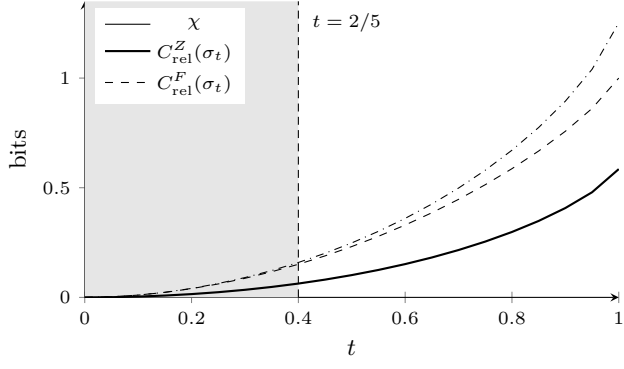

\section{Discussion and outlook}
\label{sec:discussion}

The classification establishes a direct link between the algebraic coupling structure, the induced channel class, and communication capacity. A vanishing diagonal coefficient produces an environment-orbit measure-and-prepare channel, so environmental coherence controls the classical capacity while quantum and private communication are impossible. A vanishing off-diagonal coefficient instead preserves one complete basis and produces a degradable generalized-dephasing channel. Its classical capacity is maximal for every environment, whereas the loss of quantum and private capacity is exactly the dephased entropy of the environment in the conjugate basis. The four one-zero branches are therefore completely characterized at the level of optimized Holevo information and classical, quantum, and private capacities.

The fully positive sector is the genuinely mixing regime. Weyl covariance gives an exact minimum-output-entropy formula for the one-shot Holevo information, but no single computational- or Fourier-basis coherence determines the regularized capacity. The qutrit slice demonstrates that exact one-shot formulas and sharp entanglement-breaking transitions remain accessible in this regime. Its overlap with Ref.~\cite{XiongMean2026} is restricted to the balanced pure endpoint; the depolarized interpolation and its exact EBC/NPT threshold are specific to the present work.

Two questions remain. For $d\ge5$, the $d-2$ cross-ratio classes may exhibit genuinely different additivity behavior. More generally, characterizing the environments for which minimum output entropy is additive would complete the classical-capacity landscape beyond the one-zero sector.

\begin{acknowledgments}
C.X. is supported by the National Natural Science Foundation of China (Grant No.~12201555) and the Natural Science Foundation of Hunan Province (Grant No.~2025JJ50050).
\end{acknowledgments}

\appendix

\section{Orbit-channel capacity formula}
\label{app:orbit}

\begin{lemma}
For the orbit channel in Eq.~\eqref{eq:orbit-channel}, the uniform basis ensemble is optimal and
\begin{equation*}
\chi(\Phi_\tau)=C(\Phi_\tau)=S(\mathcal T_\Gamma\tau)-S(\tau),
\end{equation*}
where $\mathcal T_\Gamma(\cdot)=|\Gamma|^{-1}\sum_gU_g(\cdot)U_g^\dagger$. Moreover, $Q(\Phi_\tau)=P(\Phi_\tau)=0$.
\end{lemma}

\begin{proof}
For a probability distribution $r$ on $\Gamma$, set
\begin{equation*}
\eta_r=\sum_gr(g)U_g\tau U_g^\dagger.
\end{equation*}
Concavity and unitary invariance of entropy give $S(\eta_r)\ge S(\tau)$. Moreover, $\mathcal T_\Gamma(\eta_r)=\mathcal T_\Gamma(\tau)$. For an arbitrary output ensemble $\{p_j,\eta_{r_j}\}$, random-unitary twirling of its average output gives
\begin{equation*}
S\!\left(\sum_jp_j\eta_{r_j}\right)
\le S(\mathcal T_\Gamma\tau),
\end{equation*}
whereas $\sum_jp_jS(\eta_{r_j})\ge S(\tau)$. Therefore
\begin{equation*}
\chi\le S(\mathcal T_\Gamma\tau)-S(\tau).
\end{equation*}
The uniform basis ensemble has outputs $U_g\tau U_g^\dagger$ and average $\mathcal T_\Gamma\tau$, so it attains the bound. The channel is measure and prepare and hence entanglement breaking. Strong additivity gives $C=\chi$, while $Q=P=0$ \cite{ShorEBC,HorodeckiShorRuskai}.
\end{proof}

\section{Entropy minimization for the qutrit slice}
\label{app:qutrit-entropy}

For $\ket\psi=\alpha\ket0+\beta\ket1+\gamma\ket2$, set $u=|\alpha|^2$, $v=|\beta|^2$, and $w=|\gamma|^2$. Direct evaluation gives
\begin{equation*}
\tau_\psi=\frac12
\begin{pmatrix}
 u+v & \beta\bar\gamma & \alpha\bar\gamma\\
 \bar\beta\gamma & u+w & \alpha\bar\beta\\
 \bar\alpha\gamma & \bar\alpha\beta & v+w
\end{pmatrix},
\end{equation*}
and hence
\begin{equation*}
\Tr\tau_\psi=1,\qquad \Tr\tau_\psi^2=\frac12.
\end{equation*}
Let $\lambda=(\lambda_1,\lambda_2,\lambda_3)$ be its spectrum and set $q_i=(1-t)/3+t\lambda_i$. To lower-bound the output entropy, it is sufficient to minimize $H(q_1,q_2,q_3)$ over the larger compact set
\begin{equation}
\lambda_i\ge0,\qquad \sum_i\lambda_i=1,\qquad \sum_i\lambda_i^2=\frac12.
\label{eq:spectral-constraints}
\end{equation}
Any minimizer lies either on the boundary or at an interior stationary point. On the boundary, Eq.~\eqref{eq:spectral-constraints} forces, up to permutation,
\begin{align*}
\lambda^{\mathrm{bdry}}&=\left(\frac12,\frac12,0\right),\\
q^{\mathrm{bdry}}(t)&=\left(\frac{2+t}{6},\frac{2+t}{6},\frac{1-t}{3}\right).
\end{align*}

For $0<t\le1$, an interior stationary point satisfies the Lagrange equations
\begin{equation*}
-\frac{t}{\ln2}\,[\ln q_i+1]=A+2B\lambda_i.
\end{equation*}
As a function of $\lambda_i$, the second derivative of the left-hand side is $t^3/(q_i^2\ln2)>0$. A strictly convex function intersects an affine function at no more than two points; hence the stationary spectrum has at most two distinct components. The three $\lambda_i$ therefore have two equal components; the constraints give
\begin{align*}
\lambda^{\mathrm{int}}&=\left(\frac23,\frac16,\frac16\right),\\
q^{\mathrm{int}}(t)&=\left(\frac{1+t}{3},\frac{2-t}{6},\frac{2-t}{6}\right).
\end{align*}
Define $F(t)=H(q^{\mathrm{int}}(t))-H(q^{\mathrm{bdry}}(t))$. Direct differentiation gives
\begin{align*}
F(0)&=0,\\
F'(t)&=\frac13\log_2\!\left[\frac{4-t^2}{4(1-t^2)}\right]>0,
\qquad 0<t<1.
\end{align*}
The boundary spectrum therefore minimizes the enlarged problem for $0<t<1$; continuity gives the same conclusion at $t=1$, and the case $t=0$ is immediate. Computational-basis inputs attain $\lambda^{\mathrm{bdry}}$, so the same spectrum is the true minimum over channel outputs.

For the curves in Fig.~\ref{fig:qutrit},
\begin{align*}
\CZ(\sigma_t)&=H(q^{\mathrm{bdry}}(t))-S(\sigma_t),\\
\CF(\sigma_t)&=H(q^{\mathrm{int}}(t))-S(\sigma_t),\\
S(\sigma_t)&=H\!\left(\frac{1+2t}{3},\frac{1-t}{3},\frac{1-t}{3}\right).
\end{align*}

\section{Partial transpose of the qutrit Choi state}
\label{app:qutrit-choi}

Let
\begin{equation*}
J_t=(\id\otimes\Lambda_{M_2,\sigma_t})(\proj{\Phi_3}),\qquad
\ket{\Phi_3}=\frac1{\sqrt3}\sum_{j=0}^2\ket{j,j},
\end{equation*}
be the normalized Choi state. A direct evaluation from $U_{M_2}$ shows that, in the ordered product basis $\{\ket{00},\ket{01},\ket{10},\ket{12},\ket{21},\ket{22},\ket{02},\ket{11},\ket{20}\}$,
\begin{equation*}
J_t^{T_B}=\frac{t+2}{18}I_6\oplus B_t,
\end{equation*}
where
\begin{equation*}
B_t=\begin{pmatrix}
(1-t)/9&t/6&t/6\\
t/6&(1-t)/9&t/6\\
t/6&t/6&(1-t)/9
\end{pmatrix}.
\end{equation*}
The symmetric vector $(1,1,1)^T$ has eigenvalue $(1+2t)/9$, while its two-dimensional orthogonal complement has eigenvalue $(2-5t)/18$. This gives the spectrum quoted in Theorem~\ref{thm:qutrit} and proves that $J_t$ is NPT exactly when $t>2/5$.

\end{document}